\documentclass[a4paper]{article}
\pdfoutput=1
\usepackage{graphics}
\usepackage{graphicx}
\usepackage{a4wide}

\usepackage{amssymb}
\usepackage{amsthm}
\usepackage{amsmath}
\usepackage{subeqnarray}
\usepackage{enumitem}

\newcommand{\cgram}{\mathcal{P}}
\newcommand{\ogram}{\mathcal{Q}}
\newcommand{\hd}{{\mathcal{H}_{2}}}
\newcommand{\hdr}[1]{{\mathcal{H}_{2,#1}}}
\newcommand{\hdo}[1]{\Vert #1 \Vert_{\mathcal{H}_{2}}}
\newcommand{\hdro}[2]{\Vert #1 \Vert_{\mathcal{H}_{2,#2}}}

\newcommand{\pare}[1]{\left ( #1 \right )}

\newcommand{\tr}[1]{\mathbf{tr}\pare{#1}}
\renewcommand{\log}[1]{\mathbf{log}\left (#1 \right)}

\newtheorem{theorem}{Theorem}
\newtheorem{remark}{Remark}
\newtheorem{definition}{Definition}

\newtheorem{proposition}{Proposition}
\newtheorem{corollary}{Corollary}
\usepackage{authblk}

\begin{document}

\title{A Spectral Expression for the Frequency-Limited $\hd$-norm}


\author[1]{Pierre Vuillemin\thanks{pierre.vuillemin@onera.fr}}
\affil[1]{ISAE \& Onera - The French Aerospace Lab, F-31055 Toulouse, France.}
\author[2]{Charles Poussot-Vassal}
\affil[2]{Onera - The French Aerospace Lab, F-31055 Toulouse, France.}
\author[1]{Daniel Alazard}
\maketitle
\begin{abstract}
In this paper, a new simple but yet efficient \emph{spectral expression of the frequency-limited $\mathcal{H}_2$-norm}, denoted $\hdr{\omega}$-norm, is introduced. The proposed new formulation requires the computation of the system eigenvalues and eigenvectors only, and provides thus an alternative to the well established Gramian-based approach. The interest of this new formulation is in three-folds: \emph{(i)} it provides a new theoretical framework for the $\hdr{\omega}$-norm-based optimization approach, such as controller synthesis, filter design and model approximation, \emph{(ii)} it improves the $\hdr{\omega}$-norm computation velocity and it applicability to models of higher dimension, and \emph{(iii)} under some conditions, it allows to handle systems with poles on the imaginary axis. Both mathematical proofs and numerical illustrations are provided to assess this new $\hdr{\omega}$-norm expression.
\end{abstract}

%



\section{Introduction}


Norms associated to Multiple Inputs Multiple Outputs (MIMO) Linear Time Invariant (LTI) dynamical models, such as $\mathcal{H}_\infty$ and $\mathcal{H}_2$-norms, take a predominant role in many control optimization problems. Indeed, these norms are often used as minimization cost for controller \cite{Doyle:1994,ZhouBook:1997,Apkarian:2006} and observer design, filter optimization or large-scale model approximation \cite{gugercin_$mathcalh_2$_2008,VuilleminSSSC:2013}. More specifically, due to practical considerations, frequency-limited versions of these norms - \emph{i.e.} their computation over a limited frequency range - are very meaningfull for engineers. Indeed, \emph{(i)} some frequencies are physically meaningless and can be regarded as uncertainties, \emph{(ii)} for vibration control, some frequencies are more specifically of interest and \emph{(iii)} in practice, the bandwidth of actuators is limited making some frequencies irrelevant for control purpose. These points are specifically true in many applicative fields such as aerospace \cite{PoussotCEP:2012} and more particularly for gust load alleviation \cite{Alazard_ASCC06,Aouf_H2_2000}, very large-scale integration systems \cite{antoulas_approximation_2005}, civilian engineering \cite{gawronski_advanced_2004}, where attention of engineers and researchers focuses on a finite frequency interval only. In these cases, a restriction of the $\hd$-norm over a bounded frequency range can be more appropriate. Therefore, its computation plays a pivotal role in many control problems, specifically when the problem dimension increases.

This paper focusses on the $\mathcal{H}_2$-norm only, and more specifically, to the new definition (and computation) of its frequency-limited version, denoted $\hdr{\omega}$, which denotes the evaluation of the $\mathcal{H}_2$-norm from 0 to an upper bound $\omega \in \mathbb{R}_+$. As grounded on the spectral information of an LTI system, the contribution of this work stands in a \emph{new spectral expression of the frequency-limited $\mathcal{H}_2$-norm}. This new surprisingly simple $\hdr{\omega}$ formulation, which only requires the eigenvalues and eigenvectors computation, \emph{(i)} exhibits many advantages with respect to the standard approaches such as applicability to large-scale models and \emph{(ii)} provides as an alternative to the Gramian-based approach, a new framework for $\hdr{\omega}$ optimization procedures, grounded on spectral informations only.



This paper is divided as follows$\,$: Section \ref{prelim} provides preliminary results on $\mathcal{H}_2$ and $\mathcal{H}_{2,\omega}$-norms definitions and computations. In Section \ref{mainres}, the main result of the paper, \emph{i.e.} a new \emph{spectral expression for the frequency-limited $\hd$-norm} is introduced. Then, Section \ref{sec3} illustrates the behaviour of the $\hdr{\omega}$-norm and deals with the numerical considerations related to this new spectral-based formulation. Finally, Section \ref{ccl} discuses the results and potential exploitation of this result.

Mathematical notations are standard$\,$: the system state is denoted $x \in \mathbb{R}^n$, the MIMO LTI dynamical system is denoted $\Sigma$ and $H(s):=C(sI_n-A)^{-1}B+D$ is its associated transfer function. The $\hdr{\omega}$-norm stands for the frequency-limited $\hd$-norm taken over the frequency interval $\left [ 0 , \omega \right ]$, $\omega \in \mathbb{R}_+$. $A^T$ and $A^*$ are the transpose and conjugate transpose of the matrix $A$, respectively. The matrix logarithm and trace operators are denoted $\textbf{logm}(.)$ and $\textbf{tr}(.)$, respectively. $\lambda(.)$ holds for the eigenvalue operator and $\lambda_i$ denotes the $i$th eigenvalue. Let $\mathbf{Re}(z)$ and $\mathbf{Im}(z)$ be the real and imaginary parts of the complex number $z=a+jb$, where $j^2 = -1$. The residual of the complex valued function $f(z)$ at $\lambda_i$ is denoted $\phi_i=\textbf{lim}_{z \rightarrow \lambda_i} (z-\lambda_i)f(z)$. The arctangent $\mathbf{atan}(z)$, logarithm $\mathbf{log}(z)$ and arccotangent $\mathbf{acot}(z)$ of the complex variable $z$ are complex functions defined in \ref{appendix}.

\section{Preliminary results on $\hd$ and $\hdr{\omega}$-norms}
\label{prelim}

Let consider $n_u$ inputs, $n_y$ outputs MIMO LTI dynamical systems of order $n$ represented as:
\begin{equation}%
\Sigma : 	\left \lbrace 
\begin{array}{l c l}%
\dot{x}(t)&=&A x(t) + B u(t)\\%
y(t)&=&Cx(t)+Du(t)%
\end{array}%
\right .
\label{FOsys}
\end{equation}
where $A\in\mathbb{R}^{n  \times n}$, $B\in \mathbb{R}^{n\times n_u}$, $C \in \mathbb{R}^{n_y \times n}$, $D \in \mathbb{R}^{n_y \times n_u}$ and its associated transfer function $H(s) = C(sI_n-A)^{-1}B+D \in \mathbb{C}^{n_y \times n_u}$, the $\hd$-norm of $\Sigma$ is defined as follows.

\begin{definition}[$\hd$-norm] 
\label{defh2}
Given a continuous MIMO LTI dynamical system $\Sigma$ described as in \eqref{FOsys}, the $\hd$-norm of $\Sigma$, denoted by $\hdo{\Sigma}$, is given as
\begin{equation}
\hdo{\Sigma}^2 := \displaystyle \tr{\frac{1}{2\pi} \int_{-\infty}^{\infty} |H(j \nu)|^2 d\nu}.
\end{equation}
\end{definition}

\noindent
When $\Sigma$ is asymptotically stable and strictly proper, then $\hdo{\Sigma}$ is finite and can be computed as
\begin{equation}
\hdo{\Sigma}^2 = \tr{B\mathcal{Q}B^T} = \tr{C^T \mathcal{P} C},
\end{equation}
where $\cgram$ and $\ogram$ are the reachability and observability Gramians solution of the following Lyapunov equations
\begin{subeqnarray}
A \cgram + \cgram  A^T + BB^T &=&0 \mbox{~,}\\
A^T\ogram + \ogram A+C^TC &=&0.
\end{subeqnarray}
Moreover, if $\Sigma$ has simple poles $\lambda_i$ ($i=1,\dots , n$), the $\mathcal{H}_2$-norm can also be computed as \cite[chap. 5]{antoulas_approximation_2005}
\begin{equation}
\hdo{\Sigma}^2 = \displaystyle \tr{\sum_{i=1}^n \phi_i H^T(-\lambda_i)},
\label{h2res}
\end{equation}
where $\lambda_i$, $\phi_i$ are the poles and residues of $H(s)$, respectively. To the authors' knowledge, this residue-based formulation is mainly used in $\hd$ optimal model reduction since it allows one to derive first-order optimality conditions directly in terms of interpolation conditions \cite{gugercin_$mathcalh_2$_2008} .

Similarly to as Definition \ref{defh2} and as firstly suggested in \cite{anderson_measure_1991}, the frequency limited $\hd$-norm can be given as follows\footnote{This norm can for instance be used for robust performance analysis \cite{masi_robust_2010} and in model approximation \cite{peterrson}.}

\begin{definition}[$\hdr{\omega}$-norm]
Given a continuous MIMO LTI dynamical system $\Sigma$ described as in \eqref{FOsys}, the frequency-limited $\hd$-norm of $\Sigma$, denoted $\hdro{\Sigma}{\omega}$, is the restriction of its $\hd$-norm over a bounded frequency range $\left [ 0, \omega \right ]$, $\omega \in \mathbb{R}^+$, and is given as
\begin{equation}
\hdro{\Sigma}{\omega}^2 := \displaystyle\tr{\frac{1}{2\pi} \int_{-\omega}^{\omega} | H(j\nu) |^2 d\nu}.
\label{eqdef}
\end{equation}
\label{defh2w}
\end{definition}

\noindent
As rooted on the Gramian-based $\hd$-norm, for a stable and strictly proper system, the resulting frequency-limited $\mathcal{H}_2$-norm can be defined as
\begin{equation}
\hdro{\Sigma}{\omega}^2 = \tr{B \ogram_\omega B^T} = \tr{C^T \cgram_\omega C}, 
\label{h2w:lyap}
\end{equation}
where $\mathcal{P}_\omega$ and  $\mathcal{Q}_\omega$ are the frequency-limited reachability and observability Gramians \cite{gawronski_advanced_2004}
\begin{subeqnarray}
\cgram _\omega &=&\displaystyle \frac{1}{2 \pi} \int_{-\omega}^{\omega} T(\nu) BB^T T^*(\nu)d\nu \slabel{flcgram}\mbox{,}\\
\ogram _\omega & =& \displaystyle \frac{1}{2 \pi} \int_{-\omega}^{\omega} T^*(\nu) C^T C T(\nu) d\nu \slabel{flogram}
\end{subeqnarray}
where $T(\nu) = \pare{j\nu I_n-A}^{-1}$. They may alternatively be obtained by solving the following two Lyapunov equations,
\begin{subeqnarray}
A \cgram _\omega + \cgram _\omega A^T + W_c(\omega) &=&0\slabel{lyap:flg1} \mbox{~,}\\
A^T\ogram _\omega + \ogram _\omega A+W_o(\omega) &=&0\slabel{lyap:flg2}
\end{subeqnarray}
where
\begin{subeqnarray}
W_c(\omega) &=&S(\omega) BB^T + BB^T S^{*}(\omega)\mbox{~,}\\
W_o(\omega) &=&S^{*}(\omega) C^TC+C^TC S(\omega) 
\end{subeqnarray}
and
\begin{equation}
\begin{array}{rcl}
S(\omega) &=& \displaystyle \frac{1}{2\pi} \int_{-\omega}^{\omega} T(\nu)d\nu \\
&=&\displaystyle \frac{j}{2 \pi} \mathbf{logm}\pare{(A+j \omega I_n)(A-j \omega I_n)^{-1}}.
\end{array}
\end{equation}

\begin{remark}[General frequency interval]
The frequency-limited $\hd$-norm can also be defined over the bounded interval $\left [\omega_1, \omega_2 \right ]$, $\omega_1$, $\omega_2\in \mathbb{R}^+$, $\omega_1<\omega_2$, as
\begin{equation}
\hdro{\Sigma}{{\left [ \omega_1, \omega_2 \right ]}}^2 = \hdro{\Sigma}{\omega_2}^2 - \hdro{\Sigma}{\omega_1}^2.
\end{equation}
\end{remark}

\section{Main result}
\label{mainres}
%
The aim of this paper is to present a new expression (up to the authors knowledge) of the $\hdr{\omega}$-norm based on the \emph{spectral informations} of the system, \emph{i.e.}  the residues and eigenvalues and of the transfer function $H(s)$, instead of the Gramian-based approach. This new simple, yet efficient, formulation allows new perspectives from a control, observer, filter design or model reduction point of view and provides an alternative to the Gramian approach for the frequency limited $\mathcal{H}_2$-norm computation of large-scale models. The main result of the paper is stated below.

\begin{theorem}[Spectral expression of the $\hdr{\omega}$-norm] Given a continuous MIMO LTI dynamical system described as in \eqref{FOsys} of degree $n$ with simple poles $\lambda_i$, $i=1,\ldots,n$. Let $\phi_i\in \mathbb{C}^{n_y\times n_u}$ denote the corresponding residues of $H(s)$ at $\lambda_i$, i.e. $\phi_i = \lim_{s\rightarrow \lambda_i} (s-\lambda_i)H(s)$, $i = 1,\ldots,n$. Suppose that the purely imaginary poles $\lambda_k^{im}$, $k = 1,\ldots,n_{im}$ ($0 \leq n_{im} \leq n$) of $H(s)$ are such that $\omega < min\left ( \vert \lambda_k^{im} \vert \right )$. Then, the frequency-limited $\mathcal{H}_2$-norm can be written as
\begin{equation}
\small
\hdro{\Sigma}{\omega}^2 = \displaystyle \sum_{i=1}^n \sum_{k=1}^n a_{i,k} 
+ \frac{\omega}{\pi} \tr{D D^T}
- \frac{2}{\pi}\sum_{i=1}^n \tr{\phi_i D^T}\mathbf{atan}\pare{\frac{\omega}{\lambda_i}} 
\label{eqmainTH}
\end{equation}
where
\begin{equation}
a_{i,k} = 
\left \lbrace \begin{array}{lr}\displaystyle \frac{2}{\pi} \tr{\frac{\phi_i\phi_k^T}{\lambda_i+\lambda_k}}\mathbf{atan}\pare{\frac{\omega}{\lambda_i}}& \text{if }\lambda_i+\lambda_k \neq 0\\
\displaystyle -\frac{1}{\pi} \tr{\frac{\omega\phi_i \phi_k^T}{\omega^2+\lambda_i \lambda_i}}& \text{otherwise.}
\end{array} \right .
\end{equation}
\label{mainTH}
\end{theorem}

\begin{proof}
Let us consider the expression of the $\hdr{\omega}$-norm of $\Sigma$, as given in Definition \ref{defh2w}
\begin{subeqnarray}
\hdro{\Sigma}{\omega}^2 &:=&\tr{\displaystyle \frac{1}{2 \pi} \int_{-\omega}^{\omega} \vert H(j\nu) \vert ^2 d\nu}\\
&=&\tr{\frac{1}{2\pi} \int_{-\omega}^{\omega} H(j \nu) H(-j\nu)^T d\nu}.
\end{subeqnarray}
As $H(s)$ has simple poles, it can be written as
\begin{equation}
H(s) = \displaystyle \sum_{i=1}^n \frac{\phi_i}{s-\lambda_i}+D.
\end{equation}
By noting $\tilde{H}(s) = \sum_{i=1}^n \frac{\phi_i}{s-\lambda_i}$, it comes that
\begin{equation}
\hdro{\Sigma}{\omega}^2 =\displaystyle \frac{1}{2\pi} \mathbf{tr} \left ( \int_{-\omega}^{\omega} DD^T+ \tilde{H}(j\nu)D^T+D\tilde{H}(-j\nu)^T +\tilde{H}(j\nu) \tilde{H}(-j\nu)^T  d\nu \right ).
\label{fourInt}
\end{equation}
Each term of this integral is then considered separately:
\begin{enumerate}[label=(\roman{*})]
\item Considering the first term, it follows that
\begin{equation}
\displaystyle \frac{1}{2 \pi} \int_{-\omega}^{\omega} DD^T d\nu = \frac{\omega}{\pi} D D^T.
\end{equation}
\item Greater attention should be turned to the following (second) integral
\begin{equation}
\displaystyle \frac{1}{2 \pi} \int_{-\omega}^{\omega} \tilde{H}(j\nu) D^T d\nu =\frac{1}{2 \pi} \sum_{i=1}^n \int_{-\omega}^{\omega} \frac{\phi_i}{j\nu-\lambda_i} D^T d\nu.
\label{inte}
\end{equation}
Indeed, if $\mathbf{Re}\pare{\lambda_i} = 0$, then $ \frac{\phi_i}{j\nu-\lambda_i}$ is integrable over $\left [ -\omega ,\omega \right ]$ if and only if $\omega < \vert \lambda_i \vert$. Hence, in the sequel, the following assumption is made:
\begin{equation}\omega <\underset{\mathbf{Re}(\lambda_i)=0}{min} \pare{|\lambda_i|}.
\label{hyp1}
\end{equation}
It implies that $ \frac{\phi_i}{j\nu-\lambda_i}$ is integrable over $\left [ -\omega ,\omega \right ]$ for all $i=1,\ldots,n$. Under this assumption and invoking formulation \eqref{atan1} of Definition \ref{def:atan} given in \ref{appendix}, which states that $\mathbf{atan}(z) = \displaystyle \frac{1}{2 j} \left [\log{1+j z}-\log{1-jz} \right ]$ for $z \neq \pm j$, it comes that
\begin{subeqnarray}
\int_{-\omega}^{\omega} \frac{\phi_i}{j\nu-\lambda_i}D^T d\nu &=&j \phi_i \left [\log{-j\omega-\lambda_i}  -\log{j\omega-\lambda_i}\right ]D^T \\
&=&-2 \phi_i \mathbf{atan}\pare{\frac{\omega}{\lambda_i}}D^T.
\label{difflog}
\end{subeqnarray}
Thus
\begin{equation}
\displaystyle \frac{1}{2 \pi} \int_{-\omega}^{\omega} \tilde{H}(j\nu) D^T d\nu =-\frac{1}{\pi} \sum_{i=1}^n \phi_i D^T \mathbf{atan}\pare{\frac{\omega}{\lambda_i}}.
\end{equation}
\item In a similar way
\begin{equation}
\displaystyle \frac{1}{2 \pi} \int_{-\omega}^{\omega} D\tilde{H}(-j\nu)^T d\nu = -\frac{1}{\pi} \sum_{i=1}^n D \phi_i^T \mathbf{atan}\pare{\frac{\omega}{\lambda_i}}.
\end{equation}
\item Regarding the last term of \eqref{fourInt}, one have
\begin{equation}
\tilde{H}(j\nu)\tilde{H}(-j\nu)^T = \sum_{i=1}^n \sum_{k=1}^n \frac{\phi_i \phi_k^T}{\pare{j\nu-\lambda_i}\pare{-j\nu-\lambda_k}}
\end{equation}
and it follows that
\begin{equation}
\displaystyle \frac{1}{2\pi} \int_{-\omega}^{\omega} \tilde{H}(j\nu)\tilde{H}(-j\nu)^T d\nu =  
\sum_{i=1}^n \sum_{k=1}^n \underbrace{\frac{1}{2 \pi} \int_{-\omega}^{\omega} \frac{\phi_i \phi_k^T}{\pare{j\nu-\lambda_i}\pare{-j\nu-\lambda_k}}d\nu}_{a_{i,k}}.
\end{equation}
From here, two cases must be considered here:
\begin{enumerate}
\item If $\lambda_i + \lambda_k \neq 0$, then
\begin{equation}
a_{i,k} = \frac{1}{2 \pi} \int_{-\omega}^{\omega} \displaystyle \frac{\phi_i \phi_k^T}{\pare{j\nu-\lambda_i}\pare{-j\nu-\lambda_k}} d\nu=\frac{1}{2 \pi} \int_{-\omega}^{\omega} \left ( \frac{p_{i,k}}{j\nu-\lambda_i}+\frac{p_{i,k}}{-j\nu-\lambda_k} r\nu\right ),
\end{equation}
with $p_{i,k} = - \frac{\phi_i \phi_k^T}{\lambda_i+ \lambda_k}$.
Because of assumption \eqref{hyp1}, each term can then be integrated as previously and one gets
\begin{equation}
a_{i,k} = \displaystyle \frac{1}{\pi} \frac{\phi_i \phi_k^T}{\lambda_i + \lambda_k} \mathbf{atan}\pare{\frac{\omega}{\lambda_i}}+\frac{1}{\pi} \frac{\phi_i \phi_k^T}{\lambda_i + \lambda_k} \mathbf{atan}\pare{\frac{\omega}{\lambda_k}}.
\end{equation}
Since $p_{i,k} = p_{k,i}^T$, the sums can be reordered as follows
\begin{equation}
a_{i,k} =  \displaystyle \frac{1}{\pi} \frac{\phi_i \phi_k^T+\phi_k \phi_i^T}{\lambda_i + \lambda_k} \mathbf{atan}\pare{\frac{\omega}{\lambda_i}}.
\end{equation}
\item If $\lambda_i+\lambda_k = 0$, then
\begin{equation}
a_{i,k} = \frac{1}{2 \pi} \int_{-\omega}^{\omega} \displaystyle \frac{\phi_i \phi_k^T}{\pare{j\nu-\lambda_i}\pare{-j\nu-\lambda_k}}d\nu =\frac{1}{2 \pi} \int_{-\omega}^{\omega} -\frac{\phi_i \phi_k^T}{ \pare{j\nu-\lambda_i}^2} d\nu.
\end{equation}
After integration, one obtains
\begin{subeqnarray}
a_{i,k} &=& \displaystyle\frac{1}{2\pi}\pare{-j  \frac{\phi_i \phi_k^T}{j\omega-\lambda_i}+j \displaystyle \frac{\phi_i \phi_k^T}{-j\omega-\lambda_i} }\\
&=&-\frac{1}{\pi} \frac{\omega \phi_i \phi_k^T}{w^2+\lambda_i\lambda_i}.
\end{subeqnarray}
\end{enumerate}
Taking the trace of each term and adding them up leads to the result of Theorem \ref{mainTH}.
\end{enumerate}
\end{proof}

\begin{corollary}[Stable and stricly proper case]
Given assumptions of Theorem \ref{mainTH}, for a continuous stable and strictly proper MIMO LTI dynamical system of degree $n$ with simple poles, the frequency-limited $\hd$-norm is given as
\begin{subeqnarray}
\small
\hdro{\Sigma}{\omega}^2& =& \displaystyle  \tr{\sum_{i=1}^n \sum_{k=1}^n \frac{2}{\pi} \frac{\phi_i\phi_k^T}{\lambda_i+\lambda_k}\mathbf{atan}\pare{\frac{\omega}{\lambda_i}}}\\
&=&\displaystyle \tr{ \sum_{i=1}^n  \phi_i H^T(-\lambda_i) \left ( -\frac{2}{\pi}\mathbf{atan}\pare{\frac{\omega}{\lambda_i}} \right )} \slabel{simpleCase}.
\end{subeqnarray}
\end{corollary}
\begin{proof}
The proof if straightforwardly obtained from \eqref{eqmainTH} with $D =0$ and by noting that as the system is stable, then $\lambda_i + \lambda_k \neq 0$, $i = 1,\ldots , n$, $k = 1,\ldots,n$.
\end{proof}

\begin{remark}[About the $\mathbf{atan}$ term]
\label{rmq:atan}
Equation \eqref{simpleCase} is very similar to the $\hd$-norm expression \eqref{h2res} presented in \cite[chap. 5]{antoulas_approximation_2005}. Indeed, the above formula involves only the new term: $-\frac{2}{\pi}\mathbf{atan}\pare{\frac{\omega}{\lambda_i}}$, which plays the role of weight for each modal contribution of the product $\phi_iH^T(-\lambda_i)$, as a function of $\omega$.
\end{remark}

\begin{remark}[Imaginary eigenvalues $\lambda_k^{im}$]
In Theorem \ref{mainTH}, assumption is made that if $H(s)$ has any imaginary poles $\lambda_k^{im}$, $k = 1,\ldots,n_{im}$ ($0 \leq n_{im} \leq n$), then $\omega$ must satisfy
\begin{equation}
\omega < min\left ( \vert \lambda_k^{im} \vert \right ).
\label{hyp}
\end{equation}
With reference to equation \eqref{difflog}, this assumption makes both formulations \eqref{atan1} and \eqref{atan2} in \ref{appendix} of the principal value of the complex arctangent, equivalent. Indeed, according to \eqref{ataneq}, the two formulations only differ when $\mathbf{Re}(z) = 0$ and $\mathbf{Im}(z) <-1$ which is never reached by $z =  \frac{\omega}{\lambda_i}$ due to hypothesis stated in relation \eqref{hyp}.
\end{remark}

\begin{remark}[Meaning of the $\hdr{\omega}$-norm for unstable systems] \label{rmq:H2unstable}
It is important to note that the $\hdr{\omega}$-norm and the $\hd$-norm are linked only for stable systems. Indeed for unstable systems, the $\hdr{\omega}$-norm is simply the integral \eqref{eqdef} and has no physical meaning.
\end{remark}

\begin{proposition}
\label{proposition}
Given a continuous and strictly proper MIMO LTI dynamical system described as in \eqref{FOsys}, the behaviour of the $\hdr{\omega}$-norm of $\Sigma$ as $\omega \to \infty$ can be summarized as follows:
\begin{enumerate}[label=(\roman{*})]
\item If $\Sigma$ is stable, then $\lim_{\omega \to \infty} \hdro{\Sigma}{\omega} = \hdo{\Sigma}$.
\item If $\Sigma$ is unstable, then $\lim_{\omega \to \infty} \hdro{\Sigma}{\omega}$ is finite (see Remark \ref{rmq:H2unstable}).
\item If $\Sigma$ has one (or more) purely imaginary pole(s), then $\lim_{\omega \to \infty} \hdro{\Sigma}{\omega} = \infty$.
\end{enumerate}
\end{proposition}

\begin{proof}
The proof of each elements is given as follows,
\begin{enumerate}[label=(\roman{*})]
\item As $\Sigma$ is stable and strictly proper, its $\hdr{\omega}$-norm is given following \eqref{simpleCase}. The limit of $\mathbf{atan}\left ( \frac{\omega}{\lambda_i} \right )$ as $\omega \to \infty$ are given by considering both the definition of $\mathbf{acot}(z)$ and its limits as $z \to 0$, $z\neq 0$ \eqref{acot}-\eqref{limacot} in \ref{appendix}. As $\mathbf{Re}(\lambda_i) <0$ for $i = 1,\ldots,n$, then
\begin{equation}
\displaystyle \lim_{\omega \to \infty} \mathbf{atan}\left (\frac{\omega}{\lambda_i} \right) = - \frac{\pi}{2} \text{ for } i=1,\ldots,n
\end{equation}
Thus $\lim_{\omega \to \infty} \hdro{\Sigma}{\omega}^2  = \tr{\sum_{i=1}^n \phi_i H^T(-\lambda_i)} = \hdo{\Sigma}^2$.\\
\item Similarly as to $(i)$, if $\lambda_i$ is an unstable pole of $H(s)$ then $\lim_{\omega \to \infty} \mathbf{atan}\left ( \frac{\omega}{\lambda}\right ) = \frac{\pi}{2}$. By separating the unstable poles $\lambda_i^{+}$, $i = 1,\ldots,n_{+}<n$ of $H(s)$ from the stable ones $\lambda_i^{-}$, $i = 1,\ldots,n_{-}<n$ such that $n_{+}+n_{-}=n$, it comes that
\begin{equation}
\lim_{\omega \to \infty} \hdro{\Sigma}{\omega}^2  = \tr{\sum_{i=1}^{n_{-}} \phi_i H^T(-\lambda_i ^{-})}- \tr{\sum_{i=1}^{n_{+}} \phi_i H^T(-\lambda_i ^{+})}.
\end{equation}
Note that if $H(s) = C(sI_n-A)^{-1}B$ has only unstable poles, then $\lim_{\omega \to \infty} \hdro{\Sigma}{\omega} = \hdo{\Sigma^-}$ where $\Sigma^{-}$ is the stable reflected system with associted tranfer function $H^-(s)=C(sI_n+A)^{-1}B$ (instead of $C(sI_n-A)^{-1}B$).\\
\item If $H(s)$ has a purely imaginary pole $\lambda_i$, then assumption \eqref{hyp} is violated as $\omega \to \infty$, thus integral \eqref{inte} is infinite.
\end{enumerate}
\end{proof}

\section{Numerical considerations \& illustration}
\label{sec3}
\subsection{Numerical considerations}
The evaluation of the $\hdr{\omega}$-norm can be computed either via the Gramian-based formulation \eqref{h2w:lyap} or via the \emph{spectral expression of Theorem \ref{mainTH}}. Given a continuous LTI dynamical system $H(s) = C(sI-A)^{-1}B +D$, the residues $\phi_i$ of the system can be computed through the right and left eigenvectors $x_i$, $y_i \in \mathbb{C}^n$ of $A$, for $i=1,\dots,n$, defined as
\begin{equation}
A x_i = \lambda_i x_i\mbox{~~, ~}y_i^* A = \lambda_i y_i^*,
\end{equation}
using the following formula \cite{Rommes_model_2008}
\begin{equation}
\phi_i =\displaystyle \frac{C x_i y_i^* B}{y_i^*x_i}\mbox{, for $i=1,\ldots,n$.}
\end{equation}

This implies that the new \emph{spectral expression of the $\hdr{\omega}$-norm} requires the computation of the eigenvalues and the corresponding eigenvectors of $A$ only. This makes the norm computation much easier than using the Gramian formulation. Indeed, while the latter requires evaluating a matrix logarithm and solving a Lyapunov equation (see Section \ref{prelim}), the spectral expression of Theorem \ref{mainTH} only requires solving an eigenvalue problem (\emph{i.e.} matrix/vector operations that can be solved for sparse models) for which many research results provide fast and accurate procedures \cite{saad_iterative_2003,Rommes_model_2008}.

Another advantage of the \emph{spectral expression of the $\hdr{\omega}$-norm}, once the eigenvalues and the residues have been computed, the $\hdr{\omega}$-norm can easily be evaluated for several values of $\omega$ with little extra numerical cost. This is not the case with the Gramian-based formulation where the last terms $W_c(\omega)$ and $W_o(\omega)$ of the Lyapunov equations \eqref{lyap:flg1} and \eqref{lyap:flg2}, which have to be solved for each value of $\omega$, involve the computation of a matrix logarithm depending on $\omega$.

In order to illustrate the computational time from a qualitative point of view, the following test has been carried on$\,$: for a fixed $\omega$ value ($\omega = 100$), the $\hdr{\omega}$-norm of several randomly generated models of dimension $n=2,\dots, 200$, has been computed using each formulation\footnote{Concerning \emph{spectral expression of Theorem \ref{mainTH}}, the eigenvalues and the eigenvectors are computed at each occurrence.} and the corresponding computation time has been measured. For each model, the norm has been evaluated $1000$ times and the mean computational time has been reported on Figure \ref{fig:time}, as a function of the model order $n$.
\begin{figure}[htbp]
\centering
\includegraphics[width = 0.7\columnwidth]{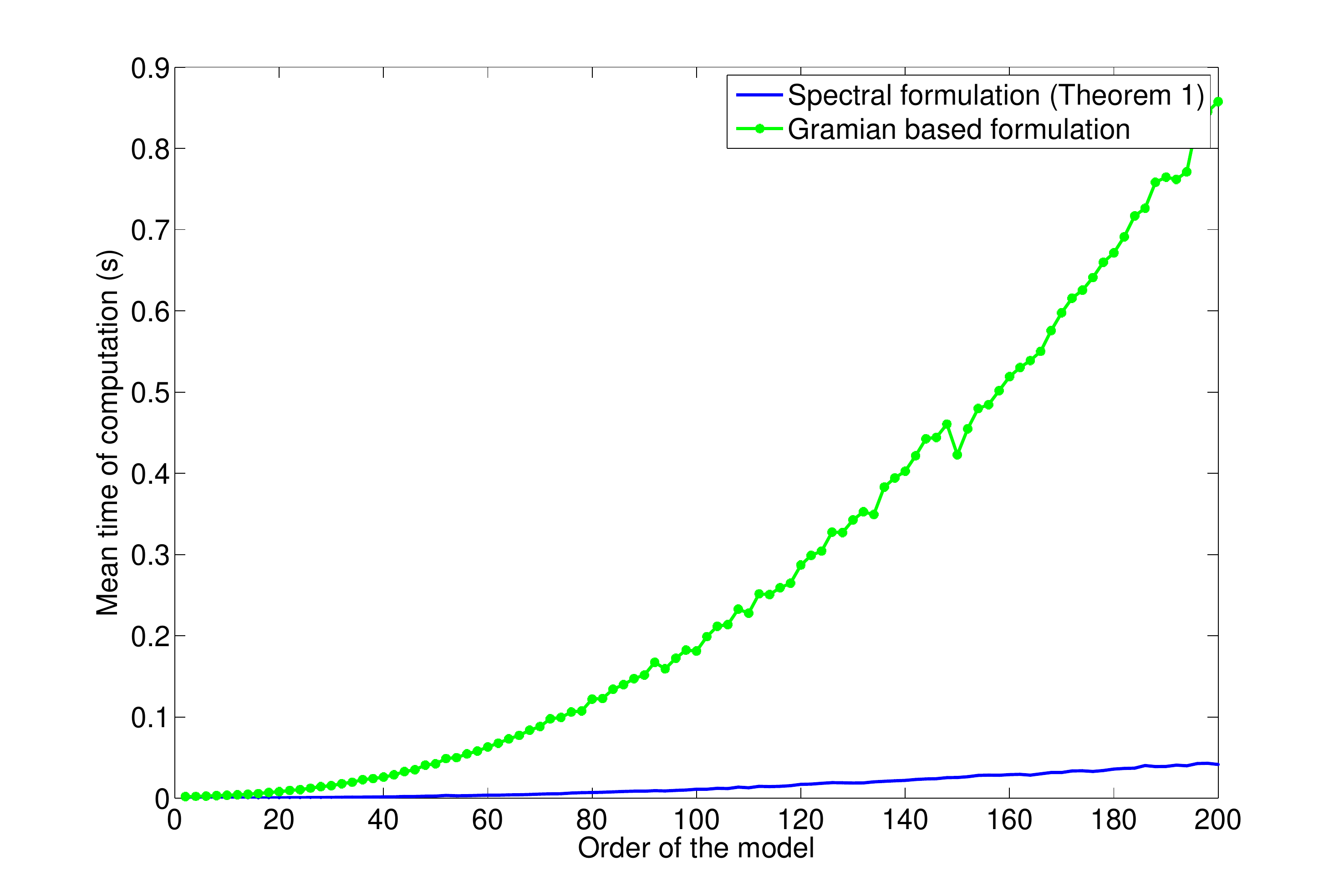}
\caption{Mean computation time of the $\hdr{\omega}$-norm for a fixed $\omega$ with respect to the order $n=2,\dots, 200$ of the randomly generated model. Solid blue line (Spectral expression of the $\hdr{\omega}$-norm, Theorem \ref{mainTH}), Solid green dotted line (Gramian-based $\hdr{\omega}$-norm).}
\label{fig:time}
\end{figure}

The standard deviations of the measures are negligible and have not been plotted. On this example, the lower complexity of the spectral formulation clearly appears in favour of the proposed \emph{spectral expression of the $\mathcal{H}_{2,\omega}$-norm}, especially for high-order models.

\subsection{Illustration of the behaviour of the $\hdr{\omega}$-norm}

Here, the Los-Angeles Hospital model provided in the COMP$l_eib$ library \cite{complib} is considered. This model is a SISO, stable and strictly proper model of order $n=48$. Figure \ref{fig:build} shows the frequency response (top) and its frequency-limited $\hd$-norm as a function of $\omega$ (bottom, solid blue line). Note that here, only the upper bound of the considered frequency interval $\left [ 0,\omega \right ]$ changes. Numerical integration, Gramian-based formulation and the \emph{spectral formulation of Theorem \ref{mainTH}} lead to the same results. This example illustrates the two following points$\,$:
\begin{enumerate}[label=(\roman{*})]
\item{As expected, for stable systems, and as $\omega$ increases, the $\hdr{\omega}$-norm increases and asymptotically tends to the $\hd$-norm (bottom, dotted red line).}
\item{The $\hdr{\omega}$-norm increases by steps. When $\omega$ crosses the abscissa of a peak in the frequency response (which corresponds to the absolute value of the poles of the system), the $\hdr{\omega}$-norm steps-up. The larger the magnitude of the peak is, the bigger the step is. This can easily be understood with the new spectral $\hdr{\omega}$-norm formula. Indeed, \eqref{simpleCase} is a sum of the product between the residues $\phi_i$ and the frequency response associated to $-\lambda_i$, weighted by $\mathbf{atan}\left ( \frac{\omega}{\lambda_i}\right )$, $i = 1,\ldots,n$ (see also Remark \ref{rmq:atan}). It is worth being noticed that $\mathbf{atan}$ have a logarithmic behaviour between each steps.}
\end{enumerate}

\begin{figure}[htbp]
\centering
\includegraphics[width = .8\columnwidth]{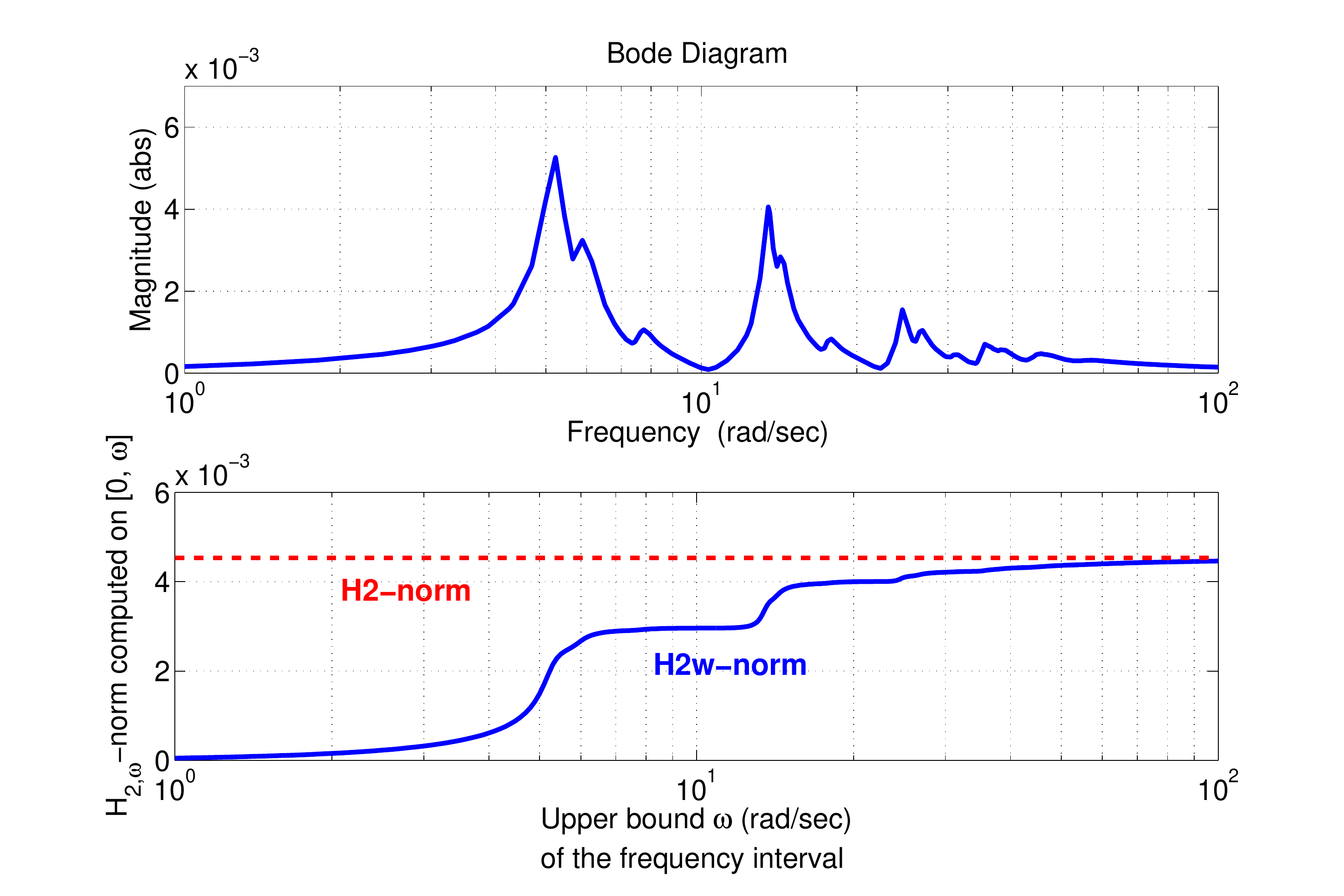}
\caption{Frequency response of the Los-Angeles Hospital model (top) and its $\hdr{\omega}$-norm (bottom) as a function of $\omega$.}
\label{fig:build}
\end{figure}

\section{Conclusions}
\label{ccl}

In this paper, a new simple yet efficient \emph{spectral-based formulation of the frequency limited $\mathcal{H}_2$-norm}, denoted $\hdr{\omega}$-norm, has been presented. For stable and strictly proper systems, this spectral formulation is very similar to the $\hd$-norm presented in \cite{antoulas_approximation_2005}. More especially the newly introduced arctangent term plays the role of weights on the contribution of each modal content and generalizes this formula. In addition to be an alternative way to formulate problems involving the $\hdr{\omega}$-norm (such as controller  design, robust analysis or in model approximation...), this \emph{spectral expression} also offers some numerical advantages over the Gramian-based formulation \eqref{h2w:lyap}. Indeed while the latter requires evaluating the matrix logarithm and solving a Lyapunov equation, the former one only requires the eigenvalues and the eigenvectors computation, which is numerically less expensive. Moreover, this new formulation makes possible the computation of the $\hdr{\omega}$-norm to large-scale systems insofar as sparsity is exploitable and there exists efficient iterative algorithms to get eigenvalues and eigenvectors (see \emph{e.g.} \cite{saad_iterative_2003,Rommes_model_2008}).
To the best of the authors' knowledge, this spectral formulation may help reconsidering some control-oriented problems from a new perspective and provide a consistent tool for many research and applicative topics.

\appendix
\section{Complex function definition and conventions}
\label{appendix}

This appendix is mainly based on \cite{haber_inverse_2011} and is dedicated to the presentation of the definition and the conventions which has been chosen for some complex functions involved in the proof of Theorem \ref{mainTH}. Given a complex number $z$, the principal value of its argument, denoted as $\mathbf{arg}(z)$, is defined such that 
\begin{equation}
-\pi < \mathbf{arg}(z) \leq \pi.
\end{equation}
The function $\mathbf{arg}(z)$ is continuous everywhere in the complex plane excepted on a branch cut along the negative real axis. Each function listed below inherit its branch cuts from the principal value of the argument.

\subsection{Complex logarithm}
\label{app:log}

\begin{definition}[Principal logarithm]
The principal value of the logarithm of $z$, denoted $\mathbf{log}(z)$, is defined for $z \neq 0$ as
\begin{equation}
\mathbf{log}(z) = \mathbf{ln} \left (|z| \right ) + j \mathbf{arg}(z),
\end{equation}
where $\mathbf{ln}(x)$ is the natural logarithm of $x \in \mathbb{R}^{*}_+$.
\end{definition}

\subsection{Complex arctangent}
\label{app:atan}

\begin{definition}[Arctangent]
\label{def:atan}
The principal value of the arctangent functions can be defined in two different ways denoted here $\mathbf{atan}_1(z)$ and $\mathbf{atan}_2(z)$. For $z \neq \pm j$,
\begin{equation}
\mathbf{atan}_1(z) = \displaystyle \frac{1}{2 j} \left [\log{1+j z}-\log{1-jz} \right ]  
\label{atan1}
\end{equation}
and
\begin{equation}
\mathbf{atan}_2(z) = \displaystyle \frac{1}{2j} \log{\frac{1+jz}{1-jz}}.
\label{atan2}
\end{equation}
\end{definition}

It can be shown that $\mathbf{atan}_1(z) = \mathbf{atan}_2(z)$ everywhere in the complex plane excepted on the branch $\left [-j,-j\infty \right )$ \cite{haber_inverse_2011}, indeed
\begin{equation}
\displaystyle \frac{1}{2j} \log{\frac{1+jz}{1-jz}} =\left \lbrace \begin{array}{ll}\displaystyle \pi +\frac{1}{2 j} \left [\log{1+j z}-\log{1-jz} \right ] & \text{ if } \mathbf{Re}(z) = 0\text{ and } \mathbf{Im}(z) <-1 \\ \displaystyle
\frac{1}{2 j} \left [\log{1+j z}-\log{1-jz} \right ] & \text{otherwise.}
\end{array}\right.
\label{ataneq}
\end{equation}

\subsection{Complex arccotangent}
\label{app:acot}

\begin{definition}[Arccotangent]
The definition of the principal value of the arccotangent of $z$, noted $\mathbf{acot}(z)$, is based on the definition of the arctangent function, indeed for $z \neq \pm j$ and $z \neq 0$,
\begin{equation}
\mathbf{acot}(z) = \displaystyle \mathbf{atan}\left (\frac{1}{z}\right).
\label{acot}
\end{equation}
\label{def:acot}
\end{definition}
Definition \ref{def:acot} leads to two formulations for the arccotangent function, $\mathbf{acot}_1(z) = \mathbf{atan}_1(\frac{1}{z})$ and $\mathbf{acot}_2(z) = \mathbf{atan}_2(\frac{1}{z})$, depending on which formulation of the arctangent is considered. These two formulations are equivalent excepted on $\left ] 0 ,j \right [$. Indeed for $z \neq \pm j$ and $z \neq 0$,
\begin{equation}
\displaystyle \frac{1}{2j} \log{\frac{z+j}{z-j}} =\left \lbrace \begin{array}{ll}\displaystyle \pi +\frac{1}{2 j} \left [\log{1+\frac{j}{z}}-\log{1-\frac{j}{z}} \right ] & \text{ if } \mathbf{Re}(z) = 0\text{ and } 0<\mathbf{Im}(z) <-1 \\ \displaystyle
\frac{1}{2 j} \left [\log{1+\frac{j}{z}}-\log{1- \frac{j}{z}} \right ] & \text{otherwise}.
\end{array}\right.
\label{equival}
\end{equation}
Note that $\mathbf{acot}_2(z)$ is defined for $z= 0$ whereas the limit of $\mathbf{acot}_1(z)$ when $z\to 0$ is not unique. Hence
\begin{equation}
\lim_{\underset{z\neq 0}{z\to 0}} \mathbf{acot}_1(z)  = \left \lbrace 
\begin{array}{rl}\frac{\pi}{2}&\text{for }\mathbf{Re}(z) >0\\
\frac{\pi}{2}&\text{for }\mathbf{Re}(z) >0\text{ and }\mathbf{Im}(z)<0\\
-\frac{\pi}{2}&\text{for }\mathbf{Re}(z) <0\\
-\frac{\pi}{2}&\text{for }\mathbf{Re}(z) <0\text{ and } \mathbf{Im}(z)>0\\
\end{array} \right .
\end{equation}
which can be simplified due to equivalence \eqref{equival} in
\begin{equation}
\lim_{\underset{z\neq 0}{z\to 0}} \mathbf{acot}_1(z)  = \left \lbrace 
\begin{array}{rl}\frac{\pi}{2}&\text{for }\mathbf{Re}(z) \geq 0\\
-\frac{\pi}{2}&\text{for }\mathbf{Re}(z) <0.
\end{array} \right .
\label{limacot}
\end{equation}
In the proof of the main result (Theorem \ref{mainTH}), formulation \eqref{atan1} is used to define the complex arctangent, but, due to theorem hypothesis, formulation \eqref{atan2} would have led to the same result.



\bibliographystyle{plain}

\bibliography{Vuillemin_Poussot-Vassal_Alazard-SpectralExpressionforthefrequency-limitedH2norm}

\begin{thebibliography}{10}

\bibitem{Alazard_ASCC06}
D.~Alazard, C.~Cumer, and F.~Delmond.
\newblock Improving flight control laws for load alleviation.
\newblock In {\em Proceedings of the $6^{th}$ Asian Control Conference}, Bali,
  Indonesia, July 2006.

\bibitem{anderson_measure_1991}
M.~R. Anderson, A.~Emami-Naeni, and J.H. Vincent.
\newblock Measures of merit for multivariable flight control.
\newblock Technical report, Systems Control Technology Inc, Palo Alto,
  California, USA, 1991.

\bibitem{antoulas_approximation_2005}
A.~C. Antoulas.
\newblock {\em Approximation of {Large-Scale} Dynamical Systems}.
\newblock Society for Industrial and Applied Mathematics, 2005.

\bibitem{Aouf_H2_2000}
N.~Aouf, B.~Boulet, and R.~Botez.
\newblock $\mathcal{H}_2$ and $\mathcal{H}_\infty$ optimal gust load
  alleviation for a flexible aircraft.
\newblock In {\em Proceedings of the 2000 American Control Conference}, pages
  1872--1876, 2000.

\bibitem{Apkarian:2006}
P.~Apkarian and D.~Noll.
\newblock {Nonsmooth $\mathcal{H}_\infty$ Synthesis}.
\newblock {\em IEEE Transaction on Automatic Control}, 51(1):71--86, January
  2006.

\bibitem{Doyle:1994}
J~C. Doyle, K.~Zhou, K.~Glover, and B.~Bodenheimer.
\newblock {Mixed $\mathcal{H}_2$ and $\mathcal{H}_\infty$ Performance
  Objectives: Optimal Control}.
\newblock {\em IEEE Transaction on Automatic Control}, 39(8):1575--1587, August
  1994.

\bibitem{gawronski_advanced_2004}
Wodek Gawronski.
\newblock {\em Advanced Structural Dynamics and Active Control of Structures}.
\newblock Springer, 1 edition, March 2004.

\bibitem{gugercin_$mathcalh_2$_2008}
S.~Gugercin, A.~C. Antoulas, and C.~Beattie.
\newblock {$\mathcal{H}_2$} model reduction for {Large-Scale} linear dynamical
  systems.
\newblock {\em {SIAM} Journal on Matrix Analysis and Applications},
  30(2):609--638, 2008.

\bibitem{haber_inverse_2011}
Howard~E. Haber.
\newblock The complex inverse trigonometric and hyperbolic functions, 2011.
\newblock University of California.

\bibitem{complib}
F.~Leibfritz and W.~Lipinski.
\newblock Description of the benchmark examples in \textit{COMP}$l_e ib$ 1.0.
\newblock Technical report, University of Trier, 2003.

\bibitem{masi_robust_2010}
A.~Masi, R.~Wallin, A.~Garulli, and A.~Hansson.
\newblock Robust finite-frequency $\mathcal{H}_2$ analysis.
\newblock In {\em Proceedings of the 49th IEEE Conferance on Decision and
  Control}, pages 6876 --6881, December 2010.

\bibitem{peterrson}
D.~Petersson and J.~L\"{o}fberg.
\newblock Model reduction using a frequency-limited $\hd$-cost.
\newblock Technical report, Linköpings universitet, 2012.

\bibitem{PoussotCEP:2012}
C.~Poussot-Vassal and C.~Roos.
\newblock {Generation of a Reduced-Order LPV/LFT Model from a Set of
  Large-Scale MIMO LTI Flexible Aircraft Models}.
\newblock {\em Control Engineering Practice}, 20(9):919--930, September 2012.

\bibitem{saad_iterative_2003}
Yousef Saad.
\newblock {\em Iterative Methods for Sparse Linear Systems, Second Edition}.
\newblock Society for Industrial and Applied Mathematics, 2 edition, April
  2003.

\bibitem{Rommes_model_2008}
W.H.A. Schilders, H.A. {Van der Vorst}, and J.~Rommes.
\newblock {\em Model Order Reduction : Theory, Research Aspects and
  Applications}, volume~13.
\newblock Springer Series Mathematics in Industry, 2008.

\bibitem{VuilleminSSSC:2013}
P.~Vuillemin, C.~Poussot-Vassal, and D.~Alazard.
\newblock {$\mathcal{H}_2$ optimal and frequency limited approximation methods
  for large-scale LTI dynamical systems}.
\newblock In {\em to appear in Proceedings of the $6$th IFAC Symposium on
  Systems Structure and Control}, Grenoble, France, February 2013.

\bibitem{ZhouBook:1997}
K.~Zhou and J~C. Doyle.
\newblock {\em Essentials Of Robust Control}.
\newblock Prentice Hall, 1997.

\end{thebibliography}

\end{document}